\documentclass[11pt]{article}
\usepackage[utf8]{inputenc}
\usepackage[margin=1.0in]{geometry}
\usepackage{amsmath,amssymb}
\usepackage{tikz}
\usepackage{hyperref}
\usepackage{todo}

\usepackage{filecontents}

\usepackage{amsfonts}
\usepackage{amsthm}
\usepackage{graphics}
\usepackage{graphicx}
\usepackage{units}
\usepackage{verbatim}
\usepackage{listings}
\usepackage{algorithm}
\usepackage{algorithmicx}
\usepackage{algcompatible}
\usepackage[noend]{algpseudocode}
\usepackage{color}
\usepackage{setspace}
\usepackage{float}
\usepackage{kantlipsum}
\usepackage{multicol}
\usepackage{fancyhdr}
\usepackage{cite}
\usepackage{qtree}
\usepackage{xspace}
\usepackage[mathscr]{euscript}
 \let\mathscr\relax% just so we can load this and rsfs
\usepackage[scr]{rsfso}
\usepackage{hyperref}
\usepackage{wrapfig}
\usepackage{bbm}
\usepackage{mathtools}
\usepackage{makecell}

\usepackage{thmtools,thm-restate}
\usepackage{cleveref}

\newtheorem{theorem}{Theorem}[section]
\newtheorem{corollary}{Corollary}[theorem]
\newtheorem{lemma}[theorem]{Lemma}
\newtheorem{definition}{Definition}[section]

\newcommand{\true}[0]{\mbox{\em true}} 
\newcommand{\false}[0]{\mbox{\em false}}

\newcommand{\go}[1]{\mbox{\sc go}_{#1}} 
\newcommand{\temp}[1]{v_{#1}} 
\newcommand{\pred}[1]{pred_{#1}}

\newcommand{\token}[0]{\texttt{token}}

\newcommand{\node}[1]{\mbox{\sc node}_{#1}}

\newcommand{\SENTINEL}[0]{\mbox{\sc sentinel}}
\newcommand{\pc}[1]{PC_{#1}}
\newcommand{\X}[0]{\mbox{\sc tail}}
\newcommand{\A}[0]{\mathcal{A}}

\newcommand{\mynode}[1]{mynode_{#1}}

\newcommand{\indic}[1]{\mathbbm{1}_{\{#1\}} } 
\newcommand{\notcached}[1]{notcached(#1)}

\newcommand{\nil}[0]{\texttt{nil}}

\newcommand{\Swap}[1]{\Call{FAS}{#1}} 
\newcommand{\CAS}[1]{\Call{CAS}{#1}}

\algdef{SE}[SUBALG]{Indent}{EndIndent}{}{\algorithmicend\ }%
\algtext*{Indent}
\algtext*{EndIndent}

\title{\bf Constant Amortized RMR Complexity \\Deterministic Abortable Mutual Exclusion Algorithm \\for CC and DSM Models}
\author{Prasad Jayanti \\ 
\href{mailto:prasad.jayanti@dartmouth.edu}{\small prasad.jayanti@dartmouth.edu} 
   \and Siddhartha Jayanti \\ 
\href{mailto:jayanti@mit.edu}{\small jayanti@mit.edu} 
}
\date{\today}

\begin{document}

\maketitle

\begin{abstract}
The {\em abortable} mutual exclusion problem was introduced by Scott and Scherer
to meet a need that arises in database and real time systems, where processes sometimes
have to abandon their attempt to acquire a mutual exclusion lock
to initiate recovery from a potential deadlock or to avoid overshooting a deadline.
Algorithms of $O(1)$ RMR complexity have been known for the standard mutual exclusion problem for both the
Cache-Coherent (CC) and Distributed Shared Memory (DSM) models of mutiprocessors,
but whether $O(1)$ RMR complexity is also achievable for abortable mutual exclusion has remained open 
for the 18 years that this problem has been investigated.

Jayanti gives a $\Theta(\log n)$ worst case RMR complexity solution for the CC and DSM models,
where $n$ is the maximum number of processes that execute the algorithm concurrently.
Giakouppis and Woelfel's algorithm, presented at PODC last year, is an $O(1)$ amortized complexity algorithm,
but it works only for the CC model, uses randomization, does not satisfy Starvation Freedom, and the $O(1)$ amortized bound holds only in expectation
and is proven for the a weak (oblivious) adversary model.

We design an algorithm that is free of these limitations:
our algorithm is deterministic, supports fast aborts (a process completes an abort in $O(1)$ steps),
has a small space complexity of $O(n)$,
requires hardware support for only the Fetch\&Store instruction,
satisfies a novely defined First Come First Served for abortable locks,
and most importantly, has $O(1)$ amortized RMR complexity for both the CC and DSM models.
Our algorithm is short and practical with fewer than a dozen lines of code, and is accompanied by
a rigorous proof of mutual exclusion through invariants and
of starvation-freedom and complexity analysis through distance and potential functions.
Thus, modulo amortization, our result answers affirmatively the long standing open question described above.

\end{abstract}

\pagebreak

\section{Introduction}

The {\em mutual exclusion problem}, proposed by Dijkstra, is a fundamental problem in concurrent computing \cite{Dijkstra}.
It calls for the design of an algorithm by which multiple asynchronous processes can compete with each other 
to acquire and release a lock so that at most one process is in possession of the lock at any time.
In over 50 years of research on this problem, the first half, surveyed in the book \cite{Raynal}, was devoted to identifying desirable properties
that an algorithm should have (e.g., lock-freedom, starvation-freedom, first-come-first-served, wait-free release, self-stabilization)
and designing algorithms to realize these properties.

The focus in the last three decades has been on designing {\em scalable} algorithms to mutual exclusion and
related problems that perform well on shared memory multiprocessors, where accessing a remotely located variable
is a few orders of magnitude slower than accessing a locally resident variable 
because of the low speed and bandwidth of the interconnection network between processors and memory modules.
Two common models of multiprocessors are considered in the literature.
In {\em distributed shared memory} (DSM) multiprocessors, the shared memory is partitioned and each processor is assigned one partition.
A shared variable $x$ resides permanently in some process's partition, but
every process $p$ can perform operations on $x$, regardless of whether $x$ resides in $p$'s partition or not.
A (read or a non-read) operation by $p$ on $x$ is counted as a {\em remote memory reference} (RMR) 
if and only if $x$ does not reside in $p$'s partition of shared memory.
In the {\em cache-coherent} (CC) model, all shared variables reside in a memory module that is remote to all processes.
Additionally, each process has a local cache where copies of shared variables can reside.
When a process $p$ reads a shared variable $x$, $x$'s copy is brought into $p$'s cache if it is not already there.
When $p$ performs a non-read operation on $x$, copies of $x$ in all caches are deleted.
Thus, a read operation by $p$ on a shared variable that is in $p$'s cache has no need to access the network,
but every other read operation and every non-read operation accesses the network and is
counted as an RMR.

{\em Remote Memory Reference (RMR) complexity},
which is the number of RMRs that a process performs in order to acquire and release the lock once,
is the common complexity measure for mutual exclusion algorithms.
The RMR complexity of an algorithm, in general, depends on $n$,
the number of processes for which the algorithm is designed,
but the ideal is to design an algorithm whose RMR complexity is $O(1)$---a constant independent of $n$.
A highlight in distributed computing research is that this ideal was achieved by
the algorithms of Graunke and Thakkar \cite{Graunke} and Anderson \cite{Anderson} for the CC model, and by Mellor-Crummey and Scott \cite{MCS}
for both the CC and DSM models.

Scott and Scherer observed that there are many systems whose demands 
are not met by the standard mutual exclusion locks described above.
Specifically, the mutual exclusion locks employed in database systems and in real time systems must support
``time out'' capability, that is, it must be possible for a process 
that waits ``too long'' to abort its attempt to acquire the lock \cite{Scott1}.
In data base systems, including Oracle's Parallel Server and IBM's DB2, the ability of a thread to abort its attempt
to acquire the lock serves the dual purpose of recovering from transaction deadlock and tolerating preemption of 
the thread that holds the lock \cite{Scott1, Scott2, Giakkoupis}.
The abort capability is also useful in real time systems to avoid overshooting a deadline.
Scott and Scherer therefore initiated research on abortable mutual exclusion locks,
which allow waiting processes to abort their attempt to acquire the lock. 
Their original algorithm \cite{Scott1} allows an aborting process to be blocked from completing its abort
by other processes, which is unacceptable.
Scott's subsequent algorithm overcomes this drawback, but has unbounded RMR complexity \cite{Scott2}.
Jayanti gives a rigorous specification of the abortable mutual exclusion problem and
presents the first algorithm of bounded RMR complexity \cite{Jayanti},
but the quintessential question of

\begin{quote}
{\sl
Is $O(1)$ RMR complexity, which is achievable for standard mutual exclusion,
also achievable for abortable mutual exclusion for both the CC and DSM models?
}
\end{quote}

\noindent
has remained open for the last 18 years.
Since Jayanti's $\Theta(\log n)$ solution for CC and DSM models 15 years ago,
there has been some progress on this question for the CC model, but no progress at all for the DSM model.
In this paper we answer this open question affirmatively in the amortized sense
by designing a single algorithm that simultaneously ensures $O(1)$ amortized RMR complexity
for both the CC and DSM models.
In the next section we describe the previous research and state our result elaborately
after providing a clear specification of the problem.
In Section 3 we present our algorithm that works on both the CC and DSM models.
We prove the correctness of our algorithm in Section 4, and state and prove RMR and space complexity bounds 
in Section 5.

\section{Problem Specification, Previous Research, and Main Result}

\subsection{Specification of abortable mutual exclusion}

In the abortable mutual exclusion problem, each process is modeled by
five sections of code---Remainder, Try, Critical, Exit, or Abort sections.
A process stays in the Remainder when it does not need the lock and, once it wants the lock,
it executes the Try section concurrently with others that are also competing for the lock.
Anytime a process is outside the Remainder section,
the environment can send an ``abort'' signal to the process
(how the environment sends this signal to a process does not concern us).
From the Try section, a process jumps either to the Critical Section (CS) or to the Abort section,
with the proviso that a process may jump to the Abort section only if it receives the ``abort'' signal from the environment
while in the Try section.
While in the CS, a process has exclusive ownership of the lock.
When it no longer needs the lock, the process gives it up by executing the Exit section to completion
and then moving back to Remainder.
If a process enters the Abort section from the Try section,
upon completing the Abort section the process moves back to Remainder.

The {\em abortable mutual exclusion problem} consists of designing the code for the Try, Exit, and Abort sections
so that the following properties are satisfied \cite{Jayanti}.

\begin{itemize}
\item 
\underline{Mutual Exclusion}: At most one process is in the CS at any time.
\item
\underline{Wait-Free Exit}: There is a bound $b$ such that each process in the Exit section 
completes that section in at most $b$ of its own steps.
\item
\underline{Wait-Free Abort}: There is a bound $b$ such that, once a process receives the abort signal
from the environment, it will enter the Remainder section in at most $b$ of its own steps.

We call this bound $b$ the {\em abort-time}.

\item
\underline{Starvation-Freedom}:
Under the assumption that no process stays in the CS forever and no process stops taking steps while in the Try, Exit, or Abort sections,
if a process in the Try section does not abort, it eventually enters the CS.
\end{itemize}

We now state another desirable property that has never been proposed or investigated earlier.
In any application, the environment sends the abort signal to a process
only when there is some urgent task that the environment needs the process to attend to.
In such a situation, we would want the process to ``quickly'' abort from its attempt to acquire the lock,
as formalized by the following property:

\begin{itemize}
\item
\underline{Fast Abort}: The abort-time is a constant that is independent of the number of 
processes for which the algorithm is designed.
\end{itemize}

We now define a novel fairness property called {\em Airline First Come First Served} (AFCFS),
which is a natural adaptation of the standard First Come First Served (FCFS) property for the abortable setting.
For intuition, imagine you are waiting to check-in in a long airline queue.
If you leave the queue to go to the restroom and return, there are two possibilities for your re-entry, both of which are reasonable.
In the first possibility, you are allowed back into your original position, as if you had never left;
in the second possibility, people behind you occupy your original position, and thus force you to go back in the queue.
Standard FCFS was not defined with aborting in mind, so it does not allow for the first possibility.
In contrast, we define our AFCFS fairness condition below to admit either of these two natural possibilities.

An {\em attempt} by a process $p$ starts when $p$ begins executing the Try section and completes
at the earliest later time when $p$ completes the Exit section or the Abort section.
Thus, the last attempt by a process maybe {\em incomplete}.
We say an attempt $a$ by a process $p$ is {\em successful} if $a$ ends in $p$ completing the Exit Section. 
%{\em unsuccessful} if $a$ ends in $p$ performing the Abort Section and not the Exit Section, and 
%{\em incomplete} if $a$ has not yet reached the Remainder Section by the end of the rundd.
Let $a_1, a_2\ldots,a_k$ be the entire sequence of attempts made by a process $p$.
A contiguous subsequence $\pi = a_i,a_{i+1},\ldots,a_j$ is called a {\em passage} if:
(1) $i = 1$ or $a_{i-1}$ is a successful attempt, 
(2)  No attempt in $a_i,\ldots,a_{j-1}$ is successful, and 
(3) $a_j$ is a successful attempt or $j = k$.
Thus, passages partition the entire attempt sequence into contiguous subsequences that end in successful attempts (or the last attempt).
As when defining standard FCFS, the Try Section consists of a bounded Doorway followed by a Waiting Room.
A passage $\pi$ by $p$ {\em $\A$-precedes} a passage $\pi'$ by $p'$ if 
$p$ completes the doorway in the last attempt of $\pi$ before $p'$ begins $\pi'$.
%We say process $p$ completes the Doorway of its passage $\pi = a_i,a_{i+1},\ldots,a_j$ when $p$ completes the Doorway in the final attempt $a_j$ of the passage.
%We define a passage $\pi = a_i,\ldots,a_j$ to have completed the doorway when the doorway is completed in the final attempt $a_j$ of the passage.
%We now state the fairness property:

\begin{itemize}
\item
\underline{Airline First Come First Served (AFCFS)}: If passage $\pi$ by process $p$ $\A$-precedes passage $\pi'$ by process $p'$,
then $p'$ does not enter the CS in $\pi'$ before $p$ enters the CS in $\pi$.
\end{itemize}

\subsection{Worst-case and amortized RMR complexity}

The {\em RMR complexity of an attempt} of a process $p$ in a run of the algorithm is the number of RMRs that $p$ performs in that passage. 
The {\em worst case RMR complexity of a run of an algorithm} is the maximum RMR complexity of an attempt in that run.
The {\em worst case RMR complexity of an algorithm} is the maximum, over all runs $R$ of the algorithm, of the RMR complexity of $R$.

The {\em amortized RMR complexity of a finite run of an algorithm} is $x/y$, where $x$ is the total number of RMRs performed in that run
by all of the processes together and $y$ is the total number of attempts initiated in that run by all of the processes together.
The {\em amortized RMR complexity of an infinite run $R$ of an algorithm} is the maximum,
over all finite prefixes $R'$ of $R$, of the amortized RMR complexity of $R'$.
The {\em amortized RMR complexity of an algorithm} is the maximum,
over all runs $R$ of the algorithm, of the amortized RMR complexity of $R$.

\subsection{Previous research}

\begin{table}[]
{\scriptsize
\begin{tabular}{|l|l|l|c|c|l|c|l|c|}
\hline
{\bf Algorithm} & {\bf Primitive} & {\bf RMRs} & {\bf WC / Amrt} & {\bf Det.} & {\bf Space} & {\bf DSM} & {\bf Fairness} & {\bf Fast Abort}\\ \hline
Scott et al. \cite{Scott1}  & $\Swap{} , \CAS{}$   & $\infty$ & WC & \checkmark   & $\infty$ & \checkmark  &  &          \\ \hline
Scott \cite{Scott2} CLH-NB  & $\Swap{} , \CAS{}$   & $\infty$ & WC & \checkmark   & $\infty$ &  \checkmark  &  &         \\ \hline
Scott \cite{Scott2} MCS-NB  & $\Swap{} , \CAS{}$   & $\infty$ & WC & \checkmark   & $\infty$ & \checkmark  &  &         \\ \hline
Jayanti \cite{Jayanti}        & $\CAS{}$   & $\Theta(\log n)$ & WC & \checkmark   & $\Theta(n)$ & \checkmark & FCFS &           \\ \hline
Lee \cite{Lee} Alg 1        & None   & $\Theta(\log n)$ & WC & \checkmark   & $\Theta(n \log n)$ & \checkmark & &          \\ \hline
Lee \cite{Lee} Alg 2        & $\Swap{} , \CAS{}$   & $\Theta(n)$ & WC & \checkmark   & $\Theta(n)$ &  &  &           \\ \hline
Lee \cite{Lee} Alg 3        & $\Swap{}$   & $\Theta(n^2)$ & WC & \checkmark   & $\infty$ &  & FCFS &         \\ \hline
Lee \cite{Lee} Alg 4        & $\Swap{}$   & $\Theta(n^2)$ & WC & \checkmark   & $\Theta(n^2)$ &  & FCFS &         \\ \hline
Woelfel et al. \cite{Woelfel} & $\CAS{}$   & $O(\frac{\log n}{\log\log n})$ & WC   &    & $\Theta(n)$ & & &         \\ \hline
Giakkoupis et al. \cite{Giakkoupis} & $\CAS{}$   & $\Theta(1)$ & Amrt   &      & $\infty$ &  & &         \\ \hline
Alon et al. \cite{Alon} & $\text{F\&A},\CAS{}$   & $O(\frac{\log n}{\log\log n})$ & WC & \checkmark   & $O(n^2)$ & & &         \\ \hline
{\em Present Work} & $\Swap{}$   & $\Theta(1)$ & Amrt  &  \checkmark   & $O(n)$ & \checkmark & AFCFS & \checkmark          \\ \hline

\end{tabular}
}
\caption{Summary of abortable locks. The columns describe: RMR complexity; whether the complexity is worst case (WC) or amortized (Amrt); whether the algorithm is Deterministic (Det.) or randomized; space complexity; whether the RMR bound holds for the DSM model; what fairness condition (if any) the algorithm satisfies; and whether the algorithm supports Fast Abort.}
\label{tab:prevwork}

\end{table}

%Giakkoupis: \makecell{Assumes oblivious adversary \\ only Live-Lock Free \\ polynomial space bound} : expected amortized
%Woelfel: Assumes oblivious adversary : Expected worst case

We list previous algorithms and their properties in Table \ref{tab:prevwork}, and now discuss the most relevant of these works.
The algorithms by Scott and Scherer \cite{Scott1}, and Scott \cite{Scott2} have unbounded RMR complexity.
Jayanti presents the first algorithm of bounded RMR complexity,
but his algorithm has a non-constant RMR of $\Theta(\log n)$, where $n$ is the number of processes that may execute the algorithm concurrently.
Attiya, Hendler, and Woefel prove that, in order to achieve sub-logarithmic RMR complexity,
one must employ randomization, amortization, or primitives other than read, write, and CAS \cite{Attiya}.

Among algorithms designed for the the CC model, the four algorithms in Lee's dissertation have non-constant RMRs of $\Theta(\log n)$ $\Theta(n)$, $\Theta(n^2)$, and $\Theta(n^2)$ \cite{Lee}.
Lee claims in passing that his second algorithm has $O(1)$ amortized RMR complexity, but he does not provide any proof to substantiate this claim.
Woelfel and Pareek \cite{Woelfel} and Alon and Morrison \cite{Alon} design randomized and deterministic algorithms, respectively,
that have sub-logarithmic, but not constant RMR complexity.
In last year's PODC, Giakkoupis and Woelfel give the only proven constant (amortized) RMR abortable lock,
but their algorithm is not deterministic and works only against a weak oblivious adversary \cite{Giakkoupis}.
Furthermore, they claim that their algorithm has polynomial space complexity,
but do not provide a bound on the degree of the polynomial.

The picture is starker for the DSM model: only two algorithms---Jayanti's \cite{Jayanti} and Lee's first of four algorithms \cite{Lee}---of bounded RMR complexity are known,
but both have logarithmic RMR complexity.

\subsection{Our result}

We present an amortized algorithm that has all of the desirable properties:
has $O(1)$ amortized RMR complexity on both CC and DSM machines, 
has $O(n)$ space complexity,
is deterministic,
satisfies the AFCFS fairness condition,
and satisfies Fast Abort.
%We present an $O(1)$ RMR complexity algorithm that is free of the limitations
%of Giakkoupis and Woelfel's algorithm and Lee's algorithms:
%our algorithm is deterministic, satisfies Fast Abort (its abort time is 6), satisfies AFCFS fairness,
%works for both the CC and DSM models, and uses only $O(n)$ space.
%As with Giakkoupis and Woelfel's and Lee's algorithms, our $O(1)$ bound is on the amortized complexity,
%and not on the worst case complexity.
Thus, modulo amortization, our result answers affirmatively the 18 year old open problem stated
at the end of the Introduction.

To the best of our knowledge, our algorithm is the first to have constant (or even sub-logarithmic) amortized RMR complexity for the DSM model,
and is the only one to satisfy Fast Abort.
Our algorithm has two more features important in practice.
First, unlike earlier algorithms which need to know $n$, the maximum number of processes 
that will execute the algorithm concurrently, and need the process names to be $1, \ldots, n$,
our algorithm works for an arbitrary number of processes of arbitrary names.
Second, our algorithm---described in full detail---is only about a dozen lines long,
making it easy to implement.

Our algorithm employs the Fetch\&Store (FAS) operation (FAS$(X, v)$ changes shared variable $X$'s value to $v$ and returns $X$'s previous value).

%In terms of techniques employed, our algorithm is a ``queue'' lock in the spirit of the MCS lock \cite{MCS}
%and Craig's lock \cite{Craig}, where waiting processes arrange themselves in a queue using the Fetch\&Store (FAS) operation
%(FAS$(X, v)$ changes shared variable $X$'s value to $v$ and returns $X$'s previous value).
%When a process $p$ aborts, $p$ marks its node $x$ in the queue as aborted.
%When $p$ enters the Try section the next time, we employ Lee's clever idea \cite{Lee} whereby $p$ tries to 
%insert itself back into the queue simply by unmarking $x$.
%However, if $x$ was already spliced out from the queue by $x$'s successor, then and only then will $p$ add itself
%anew to the end of the queue.
%Our algorithm is innovated on top of these ideas.

\section{An $O(1)$ Algorithm for CC and DSM}

In this section, we present our abortable mutual exclusion algorithm (or simply, abortable lock)
that has $O(1)$ amortized RMR complexity for both the CC and DSM models.
To help the reader understand the algorithm,
below we first present the high level ideas and only later point the reader to the actual algorithm
and our line-by-line commentary of the algorithm.

\subsection{Intuitive Description of the Main Ideas and Their Representation}

Our algorithm is essentially a queue lock: in the Try section, as processes wait to acquire the lock, they wait in a queue.
When a process enters the Try section, it adds itself to the end of the queue.
The process that is at the front of this queue is the one that enters the CS.
When a process leaves the CS, it removes itself from the queue, and lets the next process,
which is now at the front of the queue, enter the CS.
In our informal description and in the statement of our invariant,
we let $Q$ denote this abstract ``process-queue'',
$k \ge 0$ denote the number of processes in $Q$, and
$q_1, q_2, \ldots, q_k$ denote the sequence of processes in $Q$, where $q_1$ is the front process and $q_k$ is the tail process.

This abstract process-queue $Q$ is represented in the algorithm by a list of nodes.
A node is simply a single word of shared memory.
If a set $P$ consisting of $n$ processes participate in the algorithm (some of which are in the Remainder section
and the others active in Try, Critical, Exit, or Abort sections), then the total set of nodes $N$ consists of $n+1$ nodes,
of which $n$ nodes are owned by the $n$ processes and the remaining node in not owned by any process.
Thus, at any point, each process owns one node and different processes own different nodes.
However, the node that a process owns and the node that is not owned by any process change with time.
In the algorithm, each process $p$ has a local variable $\mynode{p}$ that holds the address of the node that $p$ owns.

Turning our attention back to the abstract process-queue $Q$ of length $k$,
it is represented in the algorithm by a list of $k+1$ nodes whose addresses we denote by $a_0, a_1, \ldots, a_k$,
where $a_1, a_2, \ldots, a_{k}$ are the addresses of the nodes owned by $q_1, q_2, \ldots, q_k$, respectively,
and $a_0$ is the address of the node that is not owned by any process.
Thus, for all $i \in [1, k]$, $a_i = \mynode{q_i}$.
We call the list $a_0, a_1, \ldots, a_k$ the ``node-queue'', which closely corresponds to but distinct from the process-queue $Q$.
We say $q_{i+1}$ is $q_i$'s {\em successor process} and $q_{i-1}$ is $q_i$'s {\em predecessor process}.
Similarly, we call $a_{i+1}$ is $q_i$'s {\em successor node} and $a_{i-1}$ is $q_i$'s {\em predecessor node}.
Process $q_k$ has no successor node, but $q_1$ has $a_0$ as its predecessor node.
Henceforth, we simply use the terms successor and predecessor, and let it be inferred from the context 
whether we are referring to the process or to the node.

In the algorithm, each process $p$ has another local variable $\pred{p}$ through which 
$p$ remembers the address of its predecessor node.
In particular, for all $i \in [1, k]$, $\pred{q_i} = a_{i-1}$.
There is also a shared variable $\X$, which holds $a_k$, the address of the last node in the node-queue.
When a process $p$ enters the Try section, it performs a FAS$(\X, \mynode{p})$ and stores the return value in $\pred{p}$
so as to both add itself to the end of the queue and simultaneously remember its predecessor.

The node $a_0$ is the only node that may contain a special value denoted $\token$.
Once a process $p$ enters the Try section and adds itself to the queue, 
it checks if its predecessor node contains $\token$.
If it does, $p$ knows it is $q_1$ and enters the CS.
Otherwise, $p$ busywaits until its predecessor will inform $p$ that $p$'s turn to the enter the CS has come up.
In the algorithm, there is a shared variable called $\go{p}$ for this purpose---it is on this variable that $p$ busywaits.
For $p$'s predecessor to later inform $p$ that it may enter the CS,
the predecessor needs to know the address of $\go{p}$,
so $p$ deposits $\go{p}$'s address in its predecessor node before busywaiting.

The final high level idea concerns how a process aborts.
When a process $q_i$ in the Try section wishes to abort,
it leaves its node $a_i$ intact, but marks the node as ``aborted'' by writing the address of its predecessor node into its node,
i.e., by writing $a_{i-1}$ into $*a_i$.
By doing this writing with a FAS, $q_i$ simultaneously learns the address of $\go{q_{i+1}}$,
where its successor busy-waits, and wakes up the successor.
The successor $q_{i+1}$ then reads its predecessor node $a_i$, where it sees the address $a_{i-1}$
and infers that $q_i$ has aborted, so splices out $a_i$ from the queue by writing $\nil$ in $*a_i$,
and henceforth regarding $a_{i-1}$ as its predecessor.

Suppose that after $q_i$ aborted, it decides to invoke the Try section in a bid to acquire the lock.
Since it is possible that its node is not yet spliced out of the queue by its successor,
$q_i$ tries to reclaim its old spot in the queue by simply switching the value in its node from $a_{i-1}$ to $\nil$; a similar idea was employed by Lee \cite{Lee}.
If in the process, it notices that it was already spliced out,
the process adds its node to the end of the queue as in the normal course.
Otherwise, it has happily reclaimed its old spot in queue.

\subsection{The algorithm and line-by-line commentary}

Having elaborately described the main ideas and how they are represented and implemented by the local and shared variables,
we now refer the reader to the precise algorithm (Algorithm \ref{algorithm:DSM}) and
informally explain how it works by going over the code of an arbitrary process $p$ line by line.
A note about our convention on how lines are numbered: in a single step of the execution, a process 
performs an operation on a single shared variable, but can perform any number of local actions.
Therefore, in the figure, we numbered only those lines where an operation is performed on a {\em shared} variable.

Being at Line~1 amounts to being in the Remainder section.
At Line~1, $p$ is unsure whether it had aborted its previous attempt, but if it had aborted,
$p$ knows that it would have left $\pred{p}$ in its node.
So, to reclaim its old spot in the queue (in the event that it aborted its previous attempt
and its node has not yet been spliced out of the queue by its successor),
$p$ performs a FAS on its node (Line~1).
If the FAS returns $\pred{p}$, $p$ is sure it has reclaimed its old spot in the queue and proceeds to Line~3.
Otherwise, $p$ realizes that either it didn't abort its previous attempt or its aborted node has since been spliced out,
so $p$ appends its node afresh to the queue and records its predecessor node in $\pred{p}$ (Line~2 ).
Once in the queue, $p$ performs a FAS on its predecessor node to simultaneously
inform the predecessor of the address of its busy-wait variable and learn the value $v_p$ in the predecessor node (Line~3).
If $v_p$ is $\token$, $p$ infers that it is $q_1$, the front process in the wait-queue, so terminates the Try section and proceeds to the CS.
If $v_p$ is non-$\nil$ and not the address of $p$'s busy-wait variable,
then it must be the case that the predecessor aborted and $v_p$ has the address of the predecessor's predecessor.
In this case, $p$ knows that it spliced its predecessor out of the queue,
and updates its predecessor (this shortens the queue by one node, which is crucial to proving starvation-freedom).
Having updated its predecessor, $p$ proceeds to Line~6 to check what is in store at this new predecessor.
In the remaining case (when $v_p$ is either $\nil$ or the address of its busy-wait variable),
$p$ understands it has no option but to wait until woken by its predecessor.
So, it busy-waits (Line~4) and, once woken by its predecessor,
resets $\go{p}$ to prepare it for the any busy-wait in the future (Line~5),
and then moves on to inspect the predecessor node (Line~6) to determine why the predecessor woke it up.

\begin{algorithm*}[ht]
\caption{: Abortabale Mutual Exclusion Algorithm with Amortized $O(1)$ RMR Complexity.}
\label{algorithm:DSM}

\begin{algorithmic}[1]
\begin{footnotesize}

\Statex {\bf Shared Variables}
\Statex \hspace{\algorithmicindent} $\SENTINEL$: a node, initially $\token$
\Statex \hspace{\algorithmicindent} $\X$: holds a node address, initially $\&\SENTINEL$
\Statex \hspace{\algorithmicindent} When a process $p$ joins the protocol it allocates memory for:
\Statex \hspace{\algorithmicindent} \hspace{\algorithmicindent} $\node{p}$: holds $\nil$, $\token$, or an address in shared memory,  initially $\nil$
\Statex \hspace{\algorithmicindent} \hspace{\algorithmicindent} $\go{p}$: boolean, initially $\false$. (In DSM $\go{p}$ is in $p$'s partition of shared memory)

\Statex

\Statex {\bf Local Variables}
\Statex \hspace{\algorithmicindent} For each process $p$:
\Statex \hspace{\algorithmicindent} \hspace{\algorithmicindent} $\mynode{p}$: holds the address of a $\node{}$, initially $\&\node{p}$ 
\Statex \hspace{\algorithmicindent} \hspace{\algorithmicindent} $\pred{p}$: holds the address of a $\node{}$, initially $\&\node{p}$
\Statex \hspace{\algorithmicindent} \hspace{\algorithmicindent} $\temp{p}$: holds $\nil, \token$ or an address in shared memory, arbitrarily initialized 
\end{footnotesize}
\end{algorithmic}

\begin{multicols}{2}
\begin{footnotesize}
\begin{algorithmic}[1]
\Statex \hspace{\algorithmicindent}\underline{Section \mbox{\sc Try}$(p)$}
    \State \hspace{\algorithmicindent}{\bf if} {$\Swap{*\mynode{p}, \nil} \ne \pred{p}$} {\bf then}
    \State \hspace{\algorithmicindent}\hspace{\algorithmicindent}$\pred{p} \gets \Swap{\X, \mynode{p}}$
    \State \hspace{\algorithmicindent}$\temp{p} \gets \Swap{*\pred{p}, \&\go{p}}$

    \Statex \hspace{\algorithmicindent}{\bf while} $\temp{p} \ne \token$ {\bf do} 
    \Statex \hspace{\algorithmicindent}\hspace{\algorithmicindent}{\bf if} $\temp{p} \not\in \{\nil, \&\go{p}\}$ {\bf then} $\pred{p} \gets \temp{p}$
    \State \hspace{\algorithmicindent}\hspace{\algorithmicindent}{\bf else} {\bf wait till} $\go{p} = \true$
    %\State \hspace{\algorithmicindent}\hspace{\algorithmicindent}\hspace{\algorithmicindent}{\bf wait till} $\go{p} = \true$
    \State \hspace{\algorithmicindent}\hspace{\algorithmicindent}\hspace{\algorithmicindent}$\;\;\;\go{p} \gets \false$
    \State \hspace{\algorithmicindent}\hspace{\algorithmicindent}$\temp{p} \gets \Swap{*\pred{p}, \&\go{p}}$
\end{algorithmic}

\columnbreak
\begin{algorithmic}[1]
\setcounter{ALG@line}{6}
\Statex \hspace{\algorithmicindent}\underline{Section \mbox{\sc Exit}$(p)$}
    \State \hspace{\algorithmicindent}$\temp{p} \gets \Swap{*\mynode{p}, \token}$
    \Statex \hspace{\algorithmicindent}$\mynode{p} \gets \pred{p}$
    \Statex \hspace{\algorithmicindent}{\bf if} $\temp{p} \ne \nil$ {\bf then} 
    \State \hspace{\algorithmicindent}\hspace{\algorithmicindent}$*\temp{p} \gets \true$
\Statex

\Statex \hspace{\algorithmicindent}\underline{Section \mbox{\sc Abort}$(p)$}
\State \hspace{\algorithmicindent}$\temp{p} \gets \Swap{*\pred{p}, \nil}$
\Statex \hspace{\algorithmicindent}{\bf if} $\temp{p} = \token$ {\bf then }{\bf goto }{\sc Exit} (line 7) 
%\Statex \hspace{\algorithmicindent}\hspace{\algorithmicindent}{\bf goto }{\sc Exit} (line 8)
%\State \hspace{\algorithmicindent}\hspace{\algorithmicindent}$*\pred{p} \gets \token$
\Statex \hspace{\algorithmicindent}{\bf else if} $\temp{p} \not\in \{\nil, \&\go{p}\}$ {\bf then} $\pred{p} \gets \temp{p}$

\State \hspace{\algorithmicindent}\hspace{\algorithmicindent}\hspace{0.075in}{\bf if} $(\temp{p} \gets \Swap{*\mynode{p}, \pred{p}}) \ne \nil$ %{\bf then}
    \State \hspace{\algorithmicindent}\hspace{\algorithmicindent}\hspace{\algorithmicindent}\hspace{0.075in}$*\temp{p} \gets \true$
\end{algorithmic}

\end{footnotesize}
\end{multicols}

\end{algorithm*}

Once $p$ leaves the CS, it deposits the token in its node to signal its successor that it may enter the CS (Line~7).
If the FAS operation at Line~7 returns a non-$\nil$ value,
$p$ knows that the value must be the address where the successor is busy-waiting.
So, $p$ wakes up its successor (Line~8).
Importantly, the moment $p$ deposits $\token$ in its node at Line~7,
that node becomes the new $a_0$, the token holding node that does not belong to any process,
and $p$ grabs its predecessor node (the old $a_0$) as its own node.

In our algorithm, after $p$ receives the abort signal from the environment,
it is allowed to jump to the Abort section only after performing Line~3, or Line~4, or Line~5, or Line~6.
(In particular, if the abort signal comes from the environment when $p$ is at Line~3,
it is required to execute Line~3 and all local actions associated with Line~3 before jumping to the Abort section.)
At the start of the Abort section, $p$ erases the address of its busy-wait variable from its predecessor node
because $p$ is on the way out and no longer wants to be woken by the predecessor (Line~9).
However, if $p$ observes the token in the predecessor node,
$p$ knows that it has the permission enter the CS now.
However, since $p$ wishes to abort, it will sidestep CS and proceed directly to the Exit section
(and complete its abort by executing the Exit section and returning to the Remainder from there).
Another possibility is that $v_p \neq \&\go{p}$, which means that $p$'s predecessor aborted,
in which case $p$ updates its predecessor to $v_p$, which holds $p$'s predecessor's predecessor.
At this point, $p$ marks its node as aborted by writing in it $p$'s predecessor (Line~10).
A non-nil return value would be the address where $p$'s successor is busy-waiting,
so $p$ informs the successor of its departure by setting the successor's busy-wait variable (Line~11).

\section{Proof of Correctness}

\begin{figure}

\fbox{\begin{minipage}\textwidth
There is an integer $k \ge 0$, and a sequence 
$A = a_0,\ldots,a_k$ of $k+1$ addresses of distinct nodes and a sequence 
$Q = q_1,\ldots,q_k$ of $k$ distinct processes such that
\begin{enumerate}

    \item[$I_1)$] 
    $\X = a_k$

    \item[$I_2)$] 
    $\forall i \in [1,k], \, \mynode{q_i} = a_i$
    
    \item[$I_3)$] 
    $\forall i \in [1,k], \, \pred{q_i} = a_{i-1}$
    
    \item[$I_4)$]
    $N = \{a_0\} \cup \{\mynode{p} \mid p \in P\}$
    
    \item[$I_5)$]
    $\forall p \in P,\; \pred{p} \in N$
    
    \item[$I_6)$] 
    $\forall p \in P, \, \pc{p} \in \{2, 8\} \implies p \not\in Q$ 
    
    $\forall p \in P, \, \pc{p} = 2 \implies *\mynode{p} = \nil$ 
    
    $\forall p \in P, \, \pc{p} = 8 \implies *\mynode{p} \in \{\nil, \&\go{p}\}$ 
    
    \item[$I_{7})$] 
    If $k \ge 1$, then:

    \hspace{0.4in} $\pc{q_k} \in \{3,4,5,6,7,9,10\} \implies *\mynode{q_k} = \nil$
    
    \hspace{0.4in} $\pc{q_k} \in \{1, 11\}  \implies *\mynode{q_k} = \pred{q_k}$

    If $k \ge 2$, then for all $q_i \in \{q_1,\ldots,q_{k-1}\}$:
    
    \hspace{0.4in} $\pc{q_i} \in \{3,4,5,6,7,9,10\} \implies *\mynode{p} \in \{\nil, \&\go{q_{i+1}} \}$
    
    \hspace{0.4in} $\pc{q_i} \in \{1, 11\}  \implies *\mynode{q_i} = \pred{q_i}$
        
    \item[$I_8)$] 
    $\forall p \in P, \, p \not\in Q \implies (\pc{p} \in \{1, 2, 8, 11\} \, \wedge \, *\mynode{p} \ne \pred{p})$
    
    \item[$I_{9})$] 
    If $(k = 0 \, \vee \, \pc{q_1} \ne 7)$ then $*a_0 = \token$,
    else $*a_0 \in \{\nil, \go{q_1}\}$

    \item[$I_{10})$] 
    $\forall p \in P, \, p \ne q_1 \implies \pc{p} \ne 7$

    \item[$I_{11})$] 
   $(\pc{q_1} = 4 \, \wedge \, \go{q_1} = \false) \implies \exists p \in P, \, (\pc{p} = 8 \, \wedge \, v_p = \&\go{q_1})$
   
   $\forall i \in [2,k], \, ((\pc{q_i} = 4 \, \wedge \, \go{q_i} = \false) \implies $ \\
   \mbox{ }\hspace{0.63in} $((*a_{i-1} = \&\go{q_i}) \, \vee \, ((*a_{i-1} = \pred{q_{i-1}}) \, \wedge \, (\pc{q_{i-1}} = 11) \, \wedge \, (v_{q_{i-1}} = \&\go{q_i}))))$
   
   $\forall p \in P, \pc{p} \in \{8, 11\} \implies \temp{p} \in \{\&\go{p} \mid p \in P \}$

    \item[$I_{12})$] 
    $\forall p \in P, \, \pc{p} = 5 \implies \go{p} = \true$

\end{enumerate}

Note: by $I_4$, the queue of node addresses starts with the unique node address $a_0$ that is no process $p$'s $\mynode{p}$ and by $I_1$ ends with $\X$.
This together with $I_2$ and $I_3$, implies that $k$, $A$, and $Q$ are uniquely defined.
\caption{Invariant $I = \bigwedge_{j = 1}^{14} I_{j}$ is the main invariant of Algorithm \ref{algorithm:DSM}.}
\label{fig:invariant}
\end{minipage}}

\end{figure}

% line 3 serves a dual purpose of achieving starvation freedom and registering the go_p address.
% tricky possibility of go_p being set to true by random other process.
% Importance of removing go_p when aborting

We state the invariant of Algorithm \ref{algorithm:DSM} in Figure \ref{fig:invariant}.
Below we claim the correctness of this invariant---the proof is in Appendix A.

\begin{restatable}{theorem}{invariant} 
    The statement $I$ in Figure~\ref{fig:invariant}, which is the conjunction of $I_1, I_2, \ldots, I_{12}$ 
    is an invariant of Algorithm \ref{algorithm:DSM}. Furthermore,
   the quantities $k$, $A$, and $Q$ in the invariant $I$ are unique.
    \label{thm:invariantDSM}
\end{restatable}

\begin{corollary}
    Algorithm \ref{algorithm:DSM} satisfies mutual exclusion.
\end{corollary}

\begin{proof}
    $I_{10}$ states that only process $q_1$ can be in the critical section.
\end{proof}

\begin{theorem}
    Algorithm \ref{algorithm:DSM} satisfies AFCFS where the doorway constitutes Lines 1 and 2 of the Try Section.
\end{theorem}

\begin{proof}
    In order to show this property, we prove the following stronger statement by induction:
    {\em if process $p = q_i$ in the queue and process $p' = q_j$ in the queue with $j > i$,
    then $p$ started its passage before $p'$ finished its doorway.}
    This statement is true initially since $Q$ is empty.
    The statement continues to hold inductively whenever any process leaves $Q$ from any position.
    Thus, we left with the case where a new process $p$ enters $Q$.
    Since this case occurs only if $p$ executes Line 2, and thereby just finishes the doorway,
    it is clear that every other process in $Q$ has started its passage before $p$ finished its doorway in its current passage.
    Now in conjuction with our inductive statement, we observe that by $I_{10}$, only the first process in the $Q$, $q_1$ , can be in the CS; that finishes the proof.
\end{proof}

%The next lemma makes a simple observation that the position of a process in the wait queue never gets worse with time.
%
%\begin{lemma}
%If a process $p$ is in the $i$th position in the wait queue (i.e., $p = q_i$ in the notation in which the invariant $I$ is stated)
%and $p$'s position in the queue becomes $i'$ after a process takes a step, then $i' \le i$.
%\label{pos}
%\end{lemma}

\subsection{Proof of starvation freedom}

To prove starvation freedom, we define a distance function $\delta$ that maps each process $p$ in the Try section 
to a positive integer $\delta(p)$ that represents how far away $p$ is from entering the CS.
By our definition of $\delta$, the minimum value possible for $\delta(p)$ is 1, and it is attained exactly when $p$ is in the CS.
To show that $p$ will eventually enter the CS if it does not abort,
we prove that if $\delta(p) > 1$, there is a nonempty set $\Psi(p)$ of ``promoter'' processes in the Try, Exit, or Abort sections such that
(i) if a process from $\Psi(p)$ takes the next step, $\delta(p)$ decreases, and
(ii) if a process not in $\Psi(p)$ takes the next step, $\delta(p)$ does not increase and the promoters set $\Psi(p)$ remains unchanged.
Since a process from $\Psi(p)$ must eventually take a step, $\delta(p)$ is guaranteed to eventually decrease.
By repeatedly applying this argument, we see that $\delta(p)$ eventually attains the minimum value of 1, at which point $p$ enters the CS.

Our distance function $\delta$ is based on a carefully crafted auxiliary function $f$ that maps each process $r \in Q$ to a decimal digit,
based on $r$'s program counter $\pc{r}$ and the value of its $\go{r}$ variable, as follows.

  \begin{equation}
    f(r) =
    \begin{cases*}
       3 & if $\pc{r} = 1$ \\
       2 & if $\pc{r} = 3$ \\
       8 & if $\pc{r} = 4 \, \wedge \, \go{r} = \true$ \\
       9 & if $\pc{r} = 4 \, \wedge \, \go{r} = \false$ \\
       7 & if $\pc{r} = 5$ \\
       2 & if $\pc{r} = 6$ \\
       1 & if $\pc{r} = 7$ \\
       6 & if $\pc{r} = 9$ \\
       5 & if $\pc{r} = 10$
    \end{cases*}
  \end{equation}
  
Since a process in $Q$ cannot be at Lines 2 or 8 (by $I_6$), we didn't specify $f(r)$ for $\pc{r} \in \{2,8\}$.
For a process $p$ in the Try section, we are now ready to define $p$'s distance from CS $\delta(p)$, and $p$'s 
promoters set $\Psi(p)$.

\begin{definition}[Distance function $\delta$ and Promoters set $\Psi$]
Let $p$ be a process in the Try section or CS (i.e., $\pc{p} \in \{3, 4, 5, 6, 7 \}$).
It follows from $I_8$ that $p \in Q$.
Let $i \ge 1$ be $p$'s position in $Q$ (i.e., $p = q_i$), and let $m = \min \{j \mid 1 \le j \le i, \, \pc{q_j} \neq 1 \}$.
(Since $\pc{q_i} \in \{3, 4, 5, 6, 7\}$, $m$ is well defined.)
Then:
\begin{itemize}
\item
$p$'s distance from CS, $\delta(p)$, is defined as the $i$-digit decimal number $d_1d_2 \ldots d_i$, 
where each digit $d_j$, for $1 \le j \le i$, is specified as follows:
  \begin{equation*}
    d_j =
    \begin{cases*}
       f(q_j) = 3 & if $j < m$ \\
       f(q_m) & if $j = m$ \\
      0 & if $j > m$
    \end{cases*}
  \end{equation*}
  
For example, if $i = 4$, $m = 3$, and $\pc{q_3} = 5$, then $\delta_{p} = 3370$.
  
\item
$p$'s set of promoters, $\Psi(p)$, is defined by

  \begin{equation*}
    \Psi(p) =
    \begin{cases*}
       \{r \in P \mid \pc{r} \in \{8, 11\} \, \wedge \, \temp{r} = \&\go{q_m}\} & if $\pc{q_m} = 4 \, \wedge \, \go{q_m} = \false$ \\
       \{q_m\} & otherwise
    \end{cases*}
  \end{equation*}
  
\end{itemize}
\end{definition}

\noindent
The next lemma follows easily from the definition of $\delta(p)$.

\begin{lemma}
For any process $p \in Q$ in the Try section or CS, $\delta(p) \ge 1$, and $\delta(p) = 1$ if and only if $p$ is in the CS.
\label{mindist}
\end{lemma}

The values of $\delta(p)$, $\Psi(p)$, and $\pc{p}$ can change from one configuration to the next.
So, in contexts such as the next lemma where we need to refer to these values in more than one configuration,
we add a configuration $C$ as an extra parameter and denote these values as
$\delta(p, C)$, $\Psi(p, C)$, and $\pc{p}(C)$.

\begin{restatable}{lemma}{starvationLemma} 
Suppose that a process $p \in Q$ is in the Try section in a configuration $C$ (i.e., $\pc{p}(C) \in \{3, 4, 5, 6\}$) and 
some process $\pi$ (possibly the same as $p$) takes a step from $C$.
Let $C'$ denote the configuration immediately immediately after $\pi$'s step.
\begin{enumerate}
\item
If $\pi \in \Psi(p, C)$, then $\delta(p, C') < \delta(p, C)$.
\item
If $\pi \not \in \Psi(p, C)$, then either $\delta(p, C') < \delta(p, C)$ or $(\delta(p, C') = \delta(p, C) \, \wedge \, \Psi(p, C') = \Psi(p, C))$.
\end{enumerate}
\label{reduce}
\end{restatable}

\noindent
Lemma \ref{reduce}, which is used in the next theorem, is proved in Appendix B.

\begin{theorem}[Starvation Freedom]
If a process $p \in Q$ is in the Try section (i.e., $\pc{p} \in \{3, 4, 5, 6\}$) and does not abort, it eventually enters the CS.
\label{sf}
\end{theorem}

\begin{proof}
Let $C$ be a configuration where $\pc{p} \in \{3, 4, 5, 6\}$.
Since $p$ is not in the CS in $C$, , it follows from Lemma~\ref{mindist} that $\delta(p) > 1$ in $C$.
By Lemma~\ref{reduce}, as processes take steps from $C$, $\delta(p)$ never increases.
It cannot remain the same forever because some process in $\Psi(p,C)$ eventually takes a step, 
causing $\delta(p)$ to decrease (by Lemma~\ref{reduce}).
Thus, $\delta(p)$ keeps decreasing as the execution progresses,
until eventually hitting the minimum possible value of 1, at which point $p$ is in the CS (by Lemma~\ref{mindist}).
\end{proof}

\section{Complexity Analysis}

In this section we analyze the space and RMR complexities of Algorithm \ref{algorithm:DSM}.

\begin{theorem}
    Algorithm \ref{algorithm:DSM} has a $O(n)$ space complexity, where $n$ is the total number
    of processes that join the protocol.
\end{theorem}

\begin{proof}
    The algorithm simply needs a constant number of local and shared variables per process,
    plus one $\SENTINEL{}$ node.
\end{proof}

We now wish to show that a process $p$ performs only an amortized constant number of RMRs per attempt in Algorithm \ref{algorithm:DSM} in both the CC and DSM cost models.
We analyze complexity by the potential method, by defining two different potential functions $\Phi_{CC}$ and $\Phi_{DSM}$.

We start by motivating the definition of $\Phi_{DSM}$, the simpler of the functions.
Since we are proving a constant bound,
we must show that each iteration of $p$'s while-loop (lines 4, 5, and 6) is paid for by some other action.
Since, $\go{p}$ is in $p$'s partition of memory, lines 4 and 5 have no cost in the DSM model, so we focus on line 6.
The two ways that $p$ can get to line 6 are:
(1) $\go{p}$ becomes $\true$ when $p$ is busy-waiting on line 4, and
(2) a $\Swap{}$ on either line 3 or 6 successfully removes an aborted node from the linked list.
So, we charge the executions of line 6 to processes that write $\true$ to $\go{p}$ or abort a node by writing $\pred{p}$ in $*\mynode{p}$.
This gives rise to the definition:
\[\Phi_{DSM} = \sum_{p \in P} \indic{\go{p} = true} + \indic{\pc{p} = 6} + \indic{*\mynode{p} = \pred{p}}\]
\noindent
In this definition $\indic{prop}$ is an indicator---it equals one if $prop$ is $\true$ and zero otherwise.
Note that $\Phi_{DSM}$ is a {\em proper potential function}, since it is zero in the initial configuration and always non-negative.
We now state a lemma (proved in Appendix C) that bounds the amortized cost in the DSM model of line $\ell$, $\alpha_{DSM}(\ell)$, 
for each $\ell \in [1,11]$.

\begin{restatable}{lemma}{dsmAmortizedComplexity}
    $\alpha_{DSM}(\ell) \le 1$ for lines $\ell \in [1,3] \cup \{7,9\}$,
    $\alpha_{DSM}(\ell) \le 2$ for lines $\ell \in \{8,10,11\}$, and 
    $\alpha_{DSM}(\ell) \le 0$ for $\ell \in [4,6]$.
\label{lem:DSMamortizedCosts}
\end{restatable}

The CC model is more complicated to analyze than the DSM model since:
(1) line 4 can now have a real cost if $\go{p}$ is not in $p$'s cache, and
(2) line 5 always has a real cost, since it is a write operation, and causes $\go{p}$ to become uncached.
To deal with (1), we define $\notcached{p}$ to be the predicate that $\go{p}$ is not in $p$'s cache, and add the indicator $\indic{\notcached{p}}$ to $\Phi_{CC}$.
To deal with (2), we simply multiply the weight of the indicator that $\go{p} = \true$, to pay for the
additional costs incurred on line 5.
This results in the definition:
\[\Phi_{CC} = \sum_{p \in P} 3\times\indic{\go{p} = true} + \indic{\pc{p} = 6} + \indic{*\mynode{p} = \pred{p}} + \indic{\notcached{p}} \]
\noindent
At the cost of charging one unit to a process that is newly joining the protocol, we think of $\go{p}$ as initially
residing in $p$'s cache;
so, $\Phi_{CC}$ is also a proper potential function.
We define $\alpha_{CC}(\ell)$ as the amortized cost in the CC model of line $\ell$ of Algorithm \ref{algorithm:DSM},
and state a lemma analgous to Lemma \ref{lem:DSMamortizedCosts} for the CC model (proof in Appendix C).

\begin{restatable}{lemma}{ccAmortizedComplexity}
    $\alpha_{CC}(\ell)$ is bounded by a constant for all $\ell \in [1,11] - [4,6]$, and
    $\alpha_{CC}(\ell) \le 0$ for $\ell \in [4,6]$.
\label{lem:CCamortizedCosts}
\end{restatable}

Using Lemmas \ref{lem:DSMamortizedCosts} and \ref{lem:CCamortizedCosts},
we present a unified proof of optimal amortized RMR complexity, worst-case constant Exit, and Fast Abort
in both the CC and DSM models.

\begin{theorem}
    Algorithm \ref{algorithm:DSM} has an $O(1)$ amortized RMR complexity for abortable mutual exclusion with starvation freedom and AFCFS in
    both the CC and DSM models.
    Furthermore, Algorithm \ref{algorithm:DSM} has a worst-case $O(1)$ RMR complexity for the Exit Section, uses only $O(n)$ space,
    and satisfies Fast Abort in both the CC and DSM models.
\label{thm:RMRcomplexity}
\end{theorem}

\begin{proof}
    In both the CC and DSM models, 
    all the lines (4, 5, and 6) that appear in the while-loop of Algorithm \ref{algorithm:DSM} have zero amortized cost,
    and the remaining lines have constant amortized cost by Lemma \ref{lem:CCamortizedCosts} and \ref{lem:DSMamortizedCosts}.
    Since lines that are not in the loop are executed at most once per attempt, the amortized cost of an attempt is $O(1)$
    in both the CC and DSM models.
    
    The Exit Section has only two shared memory instructions, and thus has worst-case $O(1)$ RMR complexity in both models.
    Regardless of when the environment sends an abort signal to process $p$, process $p$ can reach either the Exit or Abort
    section within three shared memory instructions, and either of these sections takes at most three more shared memory
    instructions.
    So, aborting happens within six shared instructions, and is thereby the algorithm satisfies Fast Abort.
\end{proof}

\appendix

\section{Proof of The Invariant}

\invariant*

\begin{proof}
    We will prove the invariant by induction on steps of the multiprocessor system.
    In particular, we consider what happens when a process $\pi$ executes its next step.
    
    \begin{itemize}
        \item[]
        {\em Base Case:}
        At the beginning $k = 0$ and $a_0 = \&\SENTINEL$.
        $I_1$, and $I_4$ hold true since $\X = \&\SENTINEL = a_0$.
        $I_2, I_3, I_7, I_{10}$, and $I_{11}$ hold trivially since $k = 0$.
        $I_6$ holds trivially since $\pc{p} = 1$ for every $p \in P$.
        $I_5$ holds since $\pred{p} = \mynode{p} \in N$ for every $p \in P$.
        $I_8$ holds since $\pc{p} = 1$ and $*\mynode{p} = \nil \ne \mynode{p} = \pred{p}$ for every $p \in P$.
        $I_9$ holds since $k = 0$ and $*a_0 = \SENTINEL = \token$.

        \item[]
        {\em Induction Step:}
        We assume that invariant holds in a particular configuration, and consider what happens if the next step
        is taken by some process $\pi$ executing one of the eleven possible lines.
        We use primed variables to reflect the truth after the step, and unprimed variables before when there is ambiguity.
        \begin{itemize}
            \item[]
            {\em line 1:}
            We consider two cases: $*\mynode{\pi} = \pred{\pi}$ and $*\mynode{\pi} \ne \pred{\pi}$.
            
            If $*\mynode{\pi} = \pred{\pi}$, then $\pi \in Q$ by the contrapositive of $I_8$.
            So let $\pi = q_i$.
            So, the comparison on line 1 will fail and $\pc{\pi}$ will become $3$.
            $I_7$ will hold since $*\mynode{q_i} = \nil$.
            The rest of the invariants will continue to hold since they are unaffected.
            
            If $*\mynode{\pi} \ne \pred{\pi}$, then $\pi \not\in Q$ by $I_7$.
            So, the comparison on line 1 will succeed and $\pc{\pi}$ will become $2$.
            $I_6$ and $I_8$ continue to hold since $\pi \not\in Q$ and $*\mynode{\pi} = \nil$.
            The rest of the invariants will continue to hold since they are unaffected.

            \item[]
            {\em line 2:}
            By $I_6$ we know $\pi \not\in Q$ before the line execution.
            After the execution, $q_1,\ldots,q_k$ and $a_0,\ldots,a_k$ will remain unchanged.
            The value of $k'$ will be $k+1$, with $\pi = q_{k'}$ and $\mynode{\pi} = a_{k'}$.
            Finally, $\pc{\pi}' = 3$.
            $I_1, I_2$, and $I_3$ will continue to hold by the $\Swap{}$ on line 2.
            $I_5$ continues to hold since $\pred{\pi} = a_K \in N$ by $I_4$.
            $I_7$ continues to holds for $q_{k'}$ since $*\mynode{q_{k'}} = \nil$;
            $I_7$ continues to hold for the other $q_i$'s since their program counters are unaffected,
            and we notice that $\nil \in \{\nil, \&\go{q_{i+1}} \}$.
            $I_9$ continues to hold since if $k' = 1$, the $k = 0$ before line 2,
            and otherwise it is unaffected by the line.
            The rest of the invariants will continue to hold since they are unaffected.
            
            \item[]
            {\em line 3:}
            By the contrapositive of $I_8$, we establish that $\pi \in Q$.
            So, let $\pi = q_i$ for $i \in [1,k]$.
            $I_7$ and $I_9$ imply that there are three cases:
            
            \begin{enumerate}
            \item
            Assume $i = 1$ and $*a_0 = \token$.
            In this case, $*a_0$ becomes $\&\go{q_1}$, and $\pc{\pi}' = \pc{q_1}' = 7$ after $\pi$
            notices that the while-loop condition after line 3 (which is a local action) fails.
            This immediately implies $I_9$ continues to hold.
            $I_{10}$ continues to hold since $q_1$ was unchanged by the execution of line 3.
            The rest of the invariants continue to hold since they are unaffected.
            
            \item
            Assume $i > 1$ and $*a_{i-1} \in \{\nil, \&\go{q_i}\}$.
            This means that line 3, will simply replace the initial contents of $*a_{i-1}$ with $\&\go{q_i}$,
            and $\pc{\pi}' = 4$.
            $I_{11}$ continues to hold since $(*a_{i-1} = \&\go{q_i})$.
            All invariants continue to hold since they are unaffected.
            
            \item
            Assume $i > 1$ and $*a_{i-1} = \pred{q_{i-1}}$.
            This means that line 3, will result in replacing the contents of $*a_{i-1}$ with $\&\go{q_i}$.
            This will also trigger the if-condition inside the while loop and $\pred{\pi}$ will become
            $\pred{q_{i-1}} = a_{i-2}$.
            $\pc{\pi}'$ becomes 6.
            This action removes the old $q_{i-1}$ from $Q$ by the uniqueness established through $I_1,I_2,I_3,I_4$, making $k' = k - 1$, and we rename $q_i,\ldots,q_k$ to $q'_{i-1},\ldots,q'_{k-1}$.
            $I_1, I_2$ and $I_3$ continue to hold by the renaming.
            By $I_6$ and $I_7$ we establish that the old $q_{i-1}$ had $\pc{q_{i-1}} \in \{1,11\}$;
            so we notice that $I_8$ continues to hold for $q_{i-1}$.
            The rest of the invariants continue to hold since they are unaffected.
            \end{enumerate}
            
            \item[]
            {\em line 4:}
            If $\go{\pi} = \false$, then $\pc{\pi}'$ remains at 4. 
            If $\go{\pi} = \true$, then $\pc{\pi}'$ becomes 5.
            In both cases, all invariants continue to hold as they are unaffected.
            
            \item[]
            {\em line 5:}
            $\go{\pi}'$ becomes $\false$ regardless of its initial value and $\pc{\pi}'$ becomes 6.
            All invariants continue to hold as they are unaffected.
            
            \item[]
            {\em line 6:}
            The analysis of this line is precisely identical to that of line 3.
            
            \item[]
            {\em line 7:}
            By $I_{10}$ we establish that $\pi = q_1$.
            Since, the local instruction $\mynode{\pi} \gets \pred{\pi}$ completes atomically along with the shared instruction, by $\pi \not\in Q'$ by the uniqueness of $Q'$ established by $I_1,I_2,I_3,I_4$.
            The old $\mynode{\pi} = \mynode{q_1}$ becomes the new $a'_0$, $k' = k-1$ and the old $q_2,\ldots,q_k$ and renamed to $q'_1,\ldots,q'_{k-1}$, and $I_1, I_2$ and $I_3$ continue to hold by the renaming.
            $I_4$ continues to hold since $a_0 = \mynode{\pi}'$ and $\mynode{\pi} = a_0'$; the names were simply permuted.
            $I_9$ continues to hold since $*a_0 = \token$.
            By $I_9$, $*\mynode{\pi}' = *a_0 \in \{\nil, \go{\pi}\}$, and so $I_8$ continues to hold true.

            $I_7$ shows that $*\mynode{\pi}$ could have either been $\nil$ or $\&\go{q_2}$.
            We consider three cases.
            
            \begin{enumerate}
                \item
                Assume $k = 1$.
                In this case, $I_7$ establishes that $\temp{\pi}'$ would surely have become $\nil$
                and thus $\pi$ would go to the Remainder Section ($\pc{\pi}' = 1$).
                The rest of the invariants continue to hold since they are unaffected or become trivial
                since $k' = 0$.
                
                \item
                Assume $k > 1$ and $*\mynode{\pi} = \nil$.
                $\pi$ goes to the Remainder Section ($\pc{\pi}' = 1$), since $\temp{\pi}' = \nil$.
                By $I_{11}$ on $i = 2$, we establish that $\pc{q'_1} \ne 4$ currently;
                so, $I_{11}$ continues to hold.
                The rest of the invariants continue to hold since they are unaffected.
                
                \item
                Assume $k > 1$ and $*\mynode{\pi} = \&\go{q_2} = \&\go{q'_1}$.
                This is the only case in which $\temp{\pi}' = \&\go{q'_1} \ne \nil$.
                This in turn implies that $I_{11}$ will continue to hold since $\pc{\pi}' = 8$.
                The rest of the invariants continue to hold since they are unaffected.
            \end{enumerate}
            
            \item[]
            {\em line 8:}
            $I_{11}$ is self inducting in this case, since if it were true that $(\pc{q_1} = 4 \wedge \go{q_1} = \false)$ and $\pi$ were indeed the process $p \in P$ that had $(\pc{p} = 8 \wedge \temp{p} = \&\go{q_1})$;
            then, by executing line 8, $\pi$ would ensure that $\go{q_1} = \true \ne \false$.
            The rest of the invariants continue to hold since they are unaffected.
            
            \item[]
            {\em line 9:}
            We emphasize that line 9 is almost identical to line 3, and thus the proof is also almost identical as seen below.
            By the contrapositive of $I_8$, we establish that $\pi \in Q$.
            So, let $\pi = q_i$ for $i \in [1,k]$.
            $I_7$ and $I_9$ imply that there are three cases:
            
            \begin{enumerate}
            \item
            Assume $i = 1$ and $*a_0 = \token$.
            In this case, $*a_0$ becomes $\nil$, and $\pc{\pi}' = \pc{q_1}' = 7$ after $\pi$
            notices that the if condition after line 9 (which is a local action) fails.
            This immediately implies $I_9$ continues to hold.
            $I_{10}$ continues to hold since $q_1$ was unchanged by the execution of line 9.
            The rest of the invariants continue to hold since they are unaffected.
            
            \item
            Assume $i > 1$ and $*a_{i-1} \in \{\nil, \&\go{q_i}\}$.
            This means that line 9, will simply replace the initial contents of $*a_{i-1}$ with $\nil$,
            and $\pc{\pi}' = 10$.
            All invariants continue to hold since they are unaffected.
            
            \item
            Assume $i > 1$ and $*a_{i-1} = \pred{q_{i-1}}$.
            This means that line 9, will result in replacing the contents of $*a_{i-1}$ with $\nil$.
            This will also trigger the else-if-condition and $\pred{\pi}$ will become
            $\pred{q_{i-1}} = a_{i-2}$.
            $\pc{\pi}'$ becomes 10.
            This action removes the old $q_{i-1}$ from $Q$ by the uniqueness established through $I_1,I_2,I_3,I_4$, making $k' = k - 1$, and we rename $q_i,\ldots,q_k$ to $q'_{i-1},\ldots,q'_{k-1}$.
            $I_1, I_2$ and $I_3$ continue to hold by the renaming.
            By $I_6$ and $I_7$ we establish that the old $q_{i-1}$ had $\pc{q_{i-1}} \in \{1,11\}$;
            so we notice that $I_8$ continues to hold for $q_{i-1}$.
            The rest of the invariants continue to hold since they are unaffected.
            \end{enumerate}
            
            \item[]
            {\em line 10:}
            The contrapositive of $I_8$ implies $\pi \in Q$. 
            Let $\pi = q_i$ for $i \in [1,k]$.
            $I_7$ will continue to hold since $\pc{\pi}' \in \{1, 11\}$, $\pi \in Q$, and $*\mynode{\pi} = \pred{\pi}$.
            $I_7$ shows that $*\mynode{\pi}$ could have either been $\nil$ or $\&\go{q_{i+1}}$.
            We consider three cases.
            
            \begin{enumerate}
                \item
                Assume $k = 1$.
                In this case, $I_7$ establishes that $\temp{\pi}'$ would surely have become $\nil$
                and thus $\pi$ would go to the Remainder Section ($\pc{\pi}' = 1$).
                The rest of the invariants continue to hold since they are unaffected.
                
                \item
                Assume $k > 1$ and $*\mynode{\pi} = \nil$.
                In this case, $\temp{\pi}'$ would surely have become $\nil$
                and thus $\pi$ would go to the Remainder Section ($\pc{\pi}' = 1$).
                The rest of the invariants continue to hold since they are unaffected.

                \item
                Assume $k > 1$ and $*\mynode{\pi} = \&\go{q_{i+1}}$.
                This is the only case in which $\temp{\pi}' = \&\go{q_{i+1}} \ne \nil$.
                This in turn implies that $I_{11}$ will continue to hold since 
                $(*a_i = \pred{q_i}) \wedge (\pc{q_i}' = 11) \wedge (\temp{q_i} = \&\go{q_{i-1}})$.
                The rest of the invariants continue to hold since they are unaffected.
            \end{enumerate}

            \item[]
            {\em line 11:}
            $I_{11}$ is self inducting in this case, since if it were true that $(\pc{q_i} = 4 \wedge \go{q_i} = \false)$ and $\pi$ were indeed the process $q_{i+1}$ that had $(*a_i = \pred{q_i}) \wedge (\pc{q_i}' = 11) \wedge (\temp{q_i} = \&\go{q_{i-1}})$;
            then, by executing line 11, $\pi$ would ensure that $\go{q_i} = \true \ne \false$.
            The rest of the invariants continue to hold since they are unaffected.
            
        \end{itemize}
    \end{itemize}
\end{proof}

\section{Proof of Starvation Freedom Lemma}

\starvationLemma*

\begin{proof}
Since $\pc{p}(C) \in \{3, 4, 5, 6\}$, it follows from $I_8$ that $p \in Q(C)$.
Let $p = q_i(C)$, and $\delta(p, C)$ be the $i$-digit decimal number $d_1d_2 \ldots d_i$,
where the digits $d_j$ are as defined above.
Let $m = \min \{j \mid 1 \le j \le i, \, \pc{q_j}(C) \neq 1 \}$.

We prove Part (1) of the lemma in two cases:
\begin{itemize}
\item
Suppose that $\pc{q_m}(C) = 4 \, \wedge \, \go{q_m}(C) = \false$.
Then, since $\pi \in \Psi(C)$, it follows that $\pc{\pi}(C) \in \{8, 11\} \, \wedge \, \temp{\pi}(C) = \&\go{q_m}$.
Therefore, $\pi$'s step writes $\true$ in $\go{q_m}$, making $f(q_m,C') = 8$ (note that $f(q_m,C)$ was 9).
Thus, the $m$th digit is less in $\delta(p, C')$ than in $\delta(p, C)$, while the other digits of $\delta(p)$ remain unchanged from $C$ to $C'$.
Therefore, $\delta(p, C') < \delta(p, C)$.

\item
Suppose that it is not the case that $\pc{q_m}(C) = 4 \, \wedge \, \go{q_m}(C) = \false$.
Since $\pi \in \Psi(C)$, it follows that $\pi = q_m$.

Suppose that $m = 1$ and $\pc{q_m}(C) \in \{3, 6\}$.
Then, it follows from $I_9$ that $q_1$'s step causes it jump to Line 7.
So, the first digit of $\delta(p)$ changes from $f(q_1)$ in $C$ (which is 2) to $f(q_1)$ in $C'$ (which is 1);
therefore, $\delta(p, C') < \delta(p, C)$.

Suppose that $m > 1$ and $\pc{q_m}(C) \in \{3, 6\}$.
Since $\pc{q_{m-1}} = 1$, it follows from the second part of $I_7$ that *$\mynode{q_{m-1}} = \pred{q_{m-1}}$.
Therefore, $q_{m}$'s step shortens the queue by one, causing
$p$'s position in $Q$ to change from $i$ in $C$ to $i-1$ in $C'$.
Thus, $\delta(p, C')$ has one fewer digit than $\delta(p, C)$.
Therefore, $\delta(p, C') < \delta(p, C)$.

Suppose that $\pc{q_m}(C) = 7$.
Then, $m$ must be 1 (by $I_{10}$) and $q_1$'s step causes $q_1(C)$ to be no longer in $Q$ in $C'$,
thereby causing $p$'s position in $Q$ to change from $i$ in $C$ to $i-1$ in $C'$.
Thus, $\delta(p, C')$ has one fewer digit than $\delta(p, C)$.
Therefore, $\delta(p, C') < \delta(p, C)$.

If $\pc{q_m}(C)$ is anything else (i.e., $\pc{q_m}(C) \in \{4, 5, 7, 9, 10, 11\}$),
the function $f$ is so defined that because of the changed value of $q_m$'s program counter,
$f(q_m)$ is less in $C'$ than in $C$.
Thus, the $m$th digit of $\delta(p, C')$ is less than the $m$th digit of $\delta(p, C)$
(while all other digits of $\delta(p, C')$ are respectively the same as those of $\delta(p, C)$).
Therefore, $\delta(p, C') < \delta(p, C)$.
\end{itemize}

For the proof of Part (2) of the lemma, we consider the same two cases.

\begin{itemize}
\item
Suppose that $\pc{q_m}(C) = 4 \, \wedge \, \go{q_m}(C) = \false$, and $\pi \not \in \Psi(C)$.

If $\pi = q_m$, $q_m$'s step will not change the configuration (i.e., $C' = C$);
therefore, $\delta(p, C') = \delta(p, C)$  and $\Psi(p, C') = \Psi(p, C)$.

If $\pi = q_j$ for some $j < m$, then $q_j$'s PC changes from 1 to 3.
So, by the definition of $f$, the $j$th digit of $\delta(p)$ changes from 3 in $C$ to 2 in $C'$,
while all more significant digits of $\delta(p, C')$ are respectively the same as those of $\delta(p, C)$.
Therefore, $\delta(p, C') < \delta(p, C)$.

If $\pi$ is different from all of $q_1, q_2, \ldots, q_m$, then $\delta(p, C') = \delta(p, C)$ and $\Psi(p,C') = \Psi(p,C)$.

\item
Suppose that it is not the case that $\pc{q_m}(C) = 4 \, \wedge \, \go{q_m}(C) = \false$, and $\pi \not \in \Psi(C)$.
Then $\pi \neq q_m$.
If $\pi = q_j$ for some $j < m$, then $q_j$'s PC changes from 1 to 3 and, as just argued, $\delta(p, C') < \delta(p, C)$.
If $\pi$ is different from all of $q_1, q_2, \ldots, q_m$, then $\delta(p, C') = \delta(p, C)$ and $\Psi(p,C') = \Psi(p,C)$.
\end{itemize}

\end{proof}

\section{Proof of Amortized Line Costs}

\dsmAmortizedComplexity*

\begin{proof}
    We will prove the invariant by induction on steps of the multiprocessor system.
    In particular, we consider what happens when a process $\pi$ executes its next step.
    
    \begin{itemize}
        \item[]
        {\em line 1:}
        The real cost of the line is zero or one depending on whether $*\mynode{p}$ is in $\pi$'s partition. 
        This line can only decrease the potential (by decreasing the indicator $\indic{*\mynode{p} = \pred{p}}$).
        So, amortized cost is $\alpha_{DSM}(1) \le 1$.

        \item[]
        {\em line 2:}
        The real cost of the line is one due to the \Swap{},
        and the potential function is unchanged.
        So, $\alpha_{DSM}(2) = 1$.
        
        \item[]
        {\em lines 3,7, and 9:}
        The real cost of the line is zero or one depending on whether the $\Swap{}$ is on a node in $\pi$'s partition.
        This line can only decrease the potential (by decreasing the indicator $\indic{*\mynode{p} = \pred{p}}$).
        So, $\alpha_{DSM}(3), \alpha_{DSM}(7), \alpha_{DSM}(9) \le 1$.
        
        \item[]
        {\em line 4:}
        The real cost of this line is zero since $\go{p}$ is in $p$'s partition.
        The potential change is also zero.
        So, $\alpha_{DSM}(4) = 0$.
        
        \item[]
        {\em line 5:}
        The real cost of this line is zero since $\go{p}$ is in $p$'s partition.
        When this line is executed,
        $\indic{\go{\pi} = \true}$ indicator must decrease by one due to $I_{12}$ and $\indic{\pc{\pi} = 6}$
        must increase by one.
        So, the potential change is zero.
        So, $\alpha_{DSM}(5) = 0$.
       
        \item[]
        {\em line 6:}
        The real cost of the line is zero or one depending on whether $*\pred{p}$ is in $\pi$'s partition.
        Here we have two cases.
        If $*\pred{\pi} = \mynode{p}$ for some process $p$ and $*\mynode{p} = \pred{p}$,
        then we use the indicator $\indic{*\mynode{p} = \pred{p}}$ to pay for the real cost (since $\pc{\pi}$ will
        become 6 again).
        Otherwise, $\pc{\pi}$ will end up at some other line, and we pay for the real cost using the potential 
        drop caused by the indicator $\indic{\pc{\pi} = 6}$.
        So, in either case the amortized cost of this line is $\alpha_{DSM}(6) \le 0$.
        
        \item[]
        {\em line 8, and 11:}
        The real cost of the line is zero or one depending on whether $*\temp{\pi}$ is in $\pi$'s partition.
        The potential change is at most one, since at most one indicator $\indic{\go{\pi} = \true}$ can go high.
        So, $\alpha_{DSM}(8), \alpha_{DSM}(11) \le 2$.
        
        \item[]
        {\em line 10:}
        The real cost of the line is zero or one depending on whether the $\Swap{}$ is on a node in $\pi$'s partition. 
        The potential can go up by at most one since $*\mynode{\pi} = \pred{\pi}$ after the line.
        So, $\alpha_{DSM}(10) \le 2$.
    \end{itemize}
\end{proof}

\ccAmortizedComplexity*

\begin{proof}
    We will prove the invariant by induction on steps of the multiprocessor system.
    In particular, we consider what happens when a process $\pi$ executes its next step.
    
    \begin{itemize}
        \item[]
        {\em line 1:}
        The real cost of the line is one due to the \Swap{}. 
        The change in potential is non-positive.
        So, amortized cost is $\alpha(1) \le 1$.

        \item[]
        {\em line 2:}
        The real cost of the line is one due to the \Swap{},
        The change in the potential is zero.
        So, the amortized cost is $\alpha(2) = 1$.
        
        \item[]
        {\em lines 3, 7, and 9:}
        The real cost of this line is one due to the \Swap{}.
        The change in potential is once again non-positive.
        So, the amortized cost is $\alpha_{CC}(3), \alpha_{CC}(7), \alpha_{CC}(11) \le 1$.
        
        \item[]
        {\em line 4:}
        There are two cases for this line: either $\go{\pi}$ is in $\pi$'s cache, or not.
        If $\go{\pi}$ is cached, then both the real and amortized costs are zero.
        If $\go{\pi}$ is not cached, then the real cost is one, but is cancelled out by the drop in potential caused by the fact that
        $\go{\pi}$ becomes cached.
        So, the amortized cost of this line is $\alpha_{CC}(4) \le 0$.
        
        \item[]
        {\em line 5:}
        The real cost of this line is one due to the writing of $\false$.
        Additionally, $\go{\pi}$ becomes uncached (if it was previously cached), thereby causing a one unit potential increase;
        and $\pc{\pi}$ becomes 6, causing the corresponding indicator to become one.
        However, by $I_{12}$, there is a three unit potential drop due to $\go{\pi}$ becoming $\false$.
        So, the amortized cost of this line is $\alpha_{CC}(5) = 0$.

        \item[]
        {\em line 6:}
        The real cost of the line is one due to the \Swap{}.
        Here we have two cases.
        If $*\pred{\pi} = \mynode{p}$ for some process $p$, and $*\mynode{p} = \pred{p}$,
        then we use the indicator $\indic{*\mynode{p} = \pred{p}}$ to pay for the real cost (since $\pc{\pi}$ will
        become 6 again).
        Otherwise, $\pc{\pi}$ will end up at some other line, and we pay for the real cost using the potential 
        drop caused by the indicator $\indic{\pc{\pi} = 6}$.
        So, in either case the amortized cost of this line is $\alpha_{CC}(6) \le 0$.
        
        \item[]
        {\em lines 8 and 11:}
        The real cost of this line is one for the write operation.
        Since $*\temp{\pi}$ is a $\go{}$-variable (by $I_{11}$) being set to $\true$, this line can cause a potential increase of three units.
        So, the amortized cost of this line is $\alpha_{CC}(8), \alpha_{CC}(11) \le 4$.
        
        \item[]
        {\em line 10:}
        The real cost of the line is one due to the \Swap{}.
        Since $*\mynode{\pi}$ becomes $\pred{\pi}$ due to this line,
        there is a possible potential increase of one unit.
        So, the amortized cost of this line is $\alpha_{CC}(10) \le 2$.
    \end{itemize}
\end{proof}

\end{document}